\documentclass[smallcondensed]{svjour3}   
\usepackage{graphicx}
\usepackage{amssymb}
\usepackage[cmex10]{amsmath}
\usepackage{color}
\usepackage{theorem}
\usepackage{hyperref}
\hypersetup{colorlinks=false,pdfborderstyle={/S/U/W 1},pdfborder=0 0 1}
\usepackage{url}
\usepackage{breakurl}
\usepackage{verbatim}
\usepackage{ifthen}
\usepackage[ruled,vlined,titlenumbered]{algorithm2e}
\usepackage{multirow}
\usepackage{subfigure}

\smartqed 

\newcommand{\F}[1]{\ensuremath{\mathbb{F}_{#1}}}
\newcommand{\Fq}{\F{q}}
\newcommand{\Fqn}[2]{\ensuremath{\mathbb{F}_{#1}^{#2}}}
\newcommand{\Fx}[1]{\ensuremath{\F{#1}[X]}}
\newcommand{\Fqx}{\Fx{q}}
\newcommand{\Fxy}[1]{\ensuremath{\F{#1}[X,Y]}}
\newcommand{\Fqxy}{\Fxy{q}}
\newcommand{\wdeg}[2]{\operatorname{wdeg}_{#1,#2}}
\newcommand{\yC}[2]{#1_{[#2]}} 
\newcommand{\RS}[2]{\ensuremath{\mathcal{GRS}(#1,#2)}}

\newcommand{\weight}[1]{\operatorname{weight}(#1)}
\newcommand{\defeq}{\triangleq}
\newcommand{\half}{\tfrac 1 2}
\newcommand{\diag}{\mathrm{diag}} 
\newcommand{\LT}[1]{\textnormal{\footnotesize LT}(#1)}
\newcommand{\LP}[1]{\textnormal{\footnotesize LP}(#1)}
\newcommand{\OD}[1]{\operatorname{D}(#1)}
\newcommand{\ZZ}{\mathbb Z}

\newcommand{\NN}{\mathbb N}
\newcommand{\M}[2][\empty]{
  \ifthenelse{\equal{#1}{\empty}}
    {\ensuremath{\mathcal{#2}}}
    {\ensuremath{\mathcal{#2}_{#1}}}
}
\newcommand{\Mod}[2][\empty]{
  \ifthenelse{\equal{#1}{\empty}}
    {\ensuremath{#2}}
    {\ensuremath{#2_{#1}}}
}

\newcommand{\T}[1]{^{(#1)}} 
\newcommand{\sell}[1][s,\ell]{{#1}}  
\newcommand{\sellI}{\sell[s,\ell+1]}
\newcommand{\sellII}{\sell[s+1,\ell+1]}
\newcommand{\stepI}{^\mathrm I}
\newcommand{\stepII}{^\mathrm{II}}

\def\intermediator#1{\hat{#1}}
\def\is{\intermediator s}
\def\iell{\intermediator \ell}
\def\itau{\intermediator \tau}
\def\iQ{\intermediator Q}
\newcommand{\isell}[1][\is,\iell]{\sell[#1]}
\def\var#1{\textup{\textsf{#1}}}
\DontPrintSemicolon 
\SetAlgoVlined
\LinesNumbered
\newcommand{\printalgos}[1]{%
  \begin{table}
    \begin{center}%
  \begin{minipage}{0.9\textwidth}%
      \begin{algorithm}[H]%
\DontPrintSemicolon%
\SetAlgoVlined%
\SetKwInput{KwIn}{{Input}}%
\SetKwInput{KwOut}{{Output}}%
#1%
\end{algorithm}%
\end{minipage}
\end{center}%
\end{table}}

\definecolor{orange}{rgb}{1,0.5,0}
\newtheorem{thm}{Theorem}
\newtheorem{cor}[thm]{Corollary}
\newtheorem{lem}[thm]{Lemma}
\newtheorem{exa}[thm]{Example} 
\newtheorem{defi}[thm]{Definition} 
\newenvironment{rem}{\trivlist\item\textbf{Remark}\hskip .5em}{\endtrivlist}
\renewcommand{\tilde}{\widetilde}

\begin{document}

\def\thetitle{Multi-Trial Guruswami--Sudan Decoding for Generalised Reed--Solomon Codes}
\title{\thetitle}
\titlerunning{\thetitle}

\author{Johan S.~R. Nielsen \and Alexander Zeh}
\authorrunning{J.~S.~R.~Nielsen \and A.~Zeh}
\institute{J.~S.~R.~Nielsen is with the Institute of Communications Engineering, University of Ulm, Germany.\\ 
\email{\texttt{jsrn@jsrn.dk}} \\
A.~Zeh was with the Institute of Communications Engineering, University of Ulm, Ulm, Germany and Research Center INRIA Saclay - \^{I}le-de-France, \'{E}cole Polytechnique, France and is now with the Computer Science Department, Technion---Israel Institute of Technology, Haifa, Israel.\\
\email{\texttt{alex@codingtheory.eu}}
}

\date{\today}
\maketitle

\begin{abstract}
An iterated refinement procedure for the Guruswami--Sudan list decoding algorithm for Generalised Reed--Solomon codes based on Alekhnovich's module minimisation is proposed.
The method is parametrisable and allows variants of the usual list decoding approach.
In particular, finding the list of \emph{closest} codewords within an intermediate radius can be performed with improved average-case complexity while retaining the worst-case complexity.

We provide a detailed description of the module minimisation, reanalysing the Mulders--Storjohann algorithm and drawing new connections to both Alekhnovich's algorithm and Lee--O'Sullivan's.
Furthermore, we show how to incorporate the re-encoding technique of K\"otter and Vardy into our iterative algorithm.
\end{abstract}


\section{Introduction}
Since the discovery of a polynomial-time hard-decision list decoder for Generalised Reed--Solomon (GRS)
codes by Guruswami and Sudan (GS)~\cite{sudan97,guruSudan99} in the late 1990s, much work has been done to speed up the two main parts of the algorithm: interpolation and root-finding.
Notably, for interpolation Beelen and Brander~\cite{beelenBrander} mixed the module reduction approach by Lee and O'Sullivan~\cite{leeOSullivan08} with the parametrisation of Zeh~\textit{et al.}~\cite{zeh11interpol}, and employed the fast module reduction algorithm by Alekhnovich~\cite{alekhnovich05}.
Bernstein~\cite{bernstein11} pointed out that an asymptotically faster variant can be achieved by using the reduction algorithm by Giorgi~\textit{et al.}~\cite{giorgi03}.
Very recently Chowdhury~\textit{et al.}~\cite{chowdhury_faster_2014} used fast displacement-rank linear-algebraic solvers to achieve the fastest known approach.
The GS approach was generalised by K\"otter and Vardy to a soft-decision scenario~\cite{koetter03}, and the same authors also presented a complexity-reducing re-encoding transformation~\cite{kotter_complexity_2003,koetter11}.
Cassuto~\textit{et al.}~\cite{cassuto_average_2013} and Tang~\textit{et al.}~\cite{siyun12}
proposed modified interpolation-based list decoders with reduced average-complexity.

For the root-finding step, one can employ the method of Roth and Ruckenstein~\cite{rothRuckenstein00} in a divide-and-conquer fashion, as described by Alekhnovich~\cite{alekhnovich05}.
This step then becomes an order of magnitude faster than interpolation, leaving the latter as the main target for further optimisations.

For a given GRS code, the GS algorithm has two parameters, both positive integers: the interpolation multiplicity $s$ and the list size $\ell$.
Together with the code parameters they determine the decoding radius $\tau$. To achieve a higher decoding radius for some given GRS code, one needs higher $s$ and $\ell$, and the value of these strongly influences the running time of the algorithm.

In this work, we present a novel iterative method: we first solve the interpolation problem for $s = \ell = 1$ and then iteratively refine this solution for increasing $s$ and $\ell$. In each step of our algorithm, we obtain a valid solution to the interpolation problem for these intermediate parameters.
The method builds upon that of Beelen--Brander~\cite{beelenBrander}, but a new analysis of the computational engine---Alekhnovich's module minimisation algorithm---reveals that each iteration runs faster than otherwise expected.

\def\inter#1{\hat #1}
The method therefore allows a fast \emph{multi-trial} list decoder when our aim is just to find the list of codewords with minimal distance to the received word.
At any time during the refinement process, we will have an interpolation polynomial for intermediate parameters $\inter s \leq s$, $\inter \ell \leq \ell$ yielding an intermediate decoding radius $d/2 \leq \inter \tau\leq \tau$, where $d$ is the minimum distance.
If we perform the root-finding step of the GS algorithm on this, all codewords with distance at most $\inter \tau$ from the received are returned; if there are any such words, we break computation and return those; otherwise we continue the refinement.
We can choose any number of these trials, e.g. for each possible intermediate decoding radius between $d/2$ and the target $\tau$.

Since the root-finding step of GS is less complex than the interpolation step, this multi-trial decoder will have the same asymptotic worst-case complexity as the usual GS using the Beelen--Brander interpolation~\cite{beelenBrander}.
However, its average-case complexity is better since due to the properties of the underlying channel it is more probable to have a small number of errors rather than a big one.

This contribution is structured as follows.
In the next section we give necessary preliminaries and state the GS interpolation problem for decoding GRS codes.
In Section~\ref{sec_modulemini} we give a definition and properties of minimal matrices.
We describe and reanalyse the conceptually simple Mulders--Storjohann algorithm~\cite{mulders03} for bringing matrices to this form.
We also give a fresh look at Alekhnovich's algorithm~\cite{alekhnovich05} simply as a divide-\&-conquer variant of Mulders--Storjohann, and our new analysis can carry over.
Our new iterative procedure is explained in detail in Section~\ref{sec_iterative} and the incorporation of the re-encoding transformation~\cite{koetter11} is described in Section~\ref{sec_re-encoding}.
In Section~\ref{sec_simulation} we present simulation results.

Parts of these results were presented at WCC 2013 \cite{nielsen13};
compared to that article, this version contains the incorporation of the re-encoding scheme, a full example of the algorithm, simulation results, as well as a more detailed description of the module minimisation procedure.

\section{Preliminaries} \label{sec_prelim}
\subsection{Notation}\label{ssec_notation}
Let $\Fq$ be the finite field of order $q$ and let $\Fqx$ be the polynomial ring over 
$\Fq$ with indeterminate $X$. Let $\Fqxy$ denote the polynomial ring in the variables $X$ and $Y$ 
and let $\wdeg{u}{v} X^iY^j \defeq ui + vj$ be the $(u,v)$-weighted degree of $X^iY^j$.

A vector of length $n$ is denoted by $\vec{v} = (v_0, \dots, v_{n-1})$.
If $\vec{v}$ is a vector over $\Fqx$, let $\deg \vec{v} \defeq \max_i\{ \deg v_i(X) \}$.
We introduce the leading position as 
\begin{equation} \label{def-leadingposition}
\LP{\vec{v}}= \max_i \big\{ i \mid \deg v_i(X) = \deg \vec{v} \big\}
\end{equation}
and the leading term $\LT{\vec{v}} = v_{\LP{\vec{v}}}$ is the term at this position.
An $m \times n$ matrix is denoted by $\M{V}=\| v_{i,j}\|_{i=0,j=0}^{m-1,n-1}$. The rows of such a matrix are denoted by bold lower-case letters, e.g. $\vec v_0, \ldots, \vec v_{m-1}$.
Furthermore, let the degree of such a polynomial matrix be $\deg \M{V} = \sum_{i=0}^{m-1} \deg \vec{v}_i$.
Modules are denoted by capital letters such as $M$.

\subsection{Interpolation-Based Decoding of GRS Codes} \label{ssec_ipdecoding}
Let $\alpha_0, \dots, \alpha_{n-1}$ be $n$ nonzero distinct elements 
of $\Fq$ with $n < q$ and let $w_0,\dots,w_{n-1}$ be $n$ (not necessarily distinct) nonzero elements of $\Fq$.
A GRS code \RS{n}{k} of length $n$ and dimension $k$ over $\Fq$ is given by
\begin{align} \label{eq_defRScode}
  \RS{n}{k} \defeq \big\{ (w_0 f(\alpha_{0}), \dots, w_{n-1}f(\alpha_{n-1})) : f(X) \in \Fqx,\, \deg f(X) < k \big\}.
\end{align}
GRS codes are Maximum Distance Separable (MDS) codes, i.e., their minimum Hamming distance is $d=n-k+1$.
We shortly explain the interpolation problem of GS~\cite{guruSudan99,sudan97} for list decoding GRS codes up to the Johnson radius~\cite{johnson62,bassalygo_new_1965} in the following.

\begin{thm}[Guruswami--Sudan for GRS Codes~\cite{guruSudan99,sudan97}] \label{thm_GSproblem}
Let $\vec{c} \in \RS{n}{k}$ and $f(X)$ be the corresponding information polynomial as defined in~\eqref{eq_defRScode}. Let $\vec{r} = (r_0,\dots,r_{n-1}) = \vec{c} + \vec{e}$ be a received word where $\weight{\vec e} \leq \tau$. Let $r_i^{\prime}$ denote $r_{i}/w_{i}$.

Let $Q(X,Y) \in \Fqxy$ be a nonzero polynomial that passes through the $n$ points $(\alpha_0,r_0^{\prime}),$ $ \dots, (\alpha_{n-1},r_{n-1}^{\prime})$ with multiplicity $s \geq 1$, has $Y$-degree at most $\ell$, and $\wdeg{1}{k-1} Q(X,Y) < s(n-\tau)$.
Then $(Y-f(X)) \mid Q(X,Y)$.
\end{thm}
One can easily show that a polynomial $Q(X,Y)$ that fulfills the above conditions can be constructed whenever $E(s,\ell,\tau) > 0$, where
\begin{equation}\label{eqn_Edef}
  E(s,\ell,\tau) \defeq (\ell+1)s(n-\tau) - \tbinom{\ell+1} 2 (k-1) - \tbinom{s+1} 2 n
\end{equation}
is the difference between the maximal number of coefficients of $Q(X,Y)$, and the number of homogeneous linear equations on $Q(X,Y)$ specified by the interpolation constraint.
This determines the maximal number of correctable errors, and one can show that satisfactory $s$ and $\ell$ can always be chosen whenever $\tau < n-\sqrt{n(k-1)}$.
\begin{defi}[Permissible Triples] \label{def_permiss}
An integer triple $(s,\ell,\tau) \in (\ZZ_+)^3$ is \emph{permissible} if $E(s,\ell,\tau) > 0$.
\end{defi}
We define also the decoding radius-function $\tau(s,\ell)$ as the greatest integer such that $(s,\ell,\tau(s,\ell))$ is permissible.

It is well-known that $E(s,\ell,\tau) > 0$ for $s > \ell$ implies $\tau < d/2 $, which is half the minimum distance.
Therefore, it never makes sense to consider $s > \ell$, and in the remainder we will always assume $s \leq \ell$.
Furthermore, we will also assume $s, \ell \in O(n^2)$ since this e.g.~holds for any $\tau$ for the closed-form expressions in \cite{guruSudan99}.

Let us illustrate the above. The following will be a running example throughout the article.
\begin{exa}
  \label{ex_rs164param}
  A $\RS{16}{4}$ code over \F{17} can uniquely correct $\tau_0 = (n-k)/2=6$ errors; unique decoding corresponds to $s_0 = \ell_0 = 1$ and one can confirm that $E(1, 1, 6) > 0$. 
  To attain a decoding radius $\tau_1 = 7$, one can choose $s_1 = 1$ and $\ell_1 = 2$ in order to obtain a permissible triple.
  Also $(1, 3, 7)$ is permissible, though less interesting since it does not give improved decoding radius.
  However, one finds $(2, 4, 8)$ and $(28, 64, 9)$ are permissible.
  Since $n-\sqrt{n(k-1)} < 10 $, there are no permissible triples for greater decoding radii.
\end{exa}

\subsection{Module Reformulation of Guruswami--Sudan} \label{ssec_modreform}
Let $\Mod[\sell]{M} \subset \Fqxy$ denote the space of all bivariate polynomials passing through the points $(\alpha_0,r_0^{\prime}),\ldots,(\alpha_{n-1},r_{n-1}^{\prime})$ with multiplicity $s$ and $Y$-degree at most $\ell$.
We are searching for an element of $\Mod[\sell]{M}$ with low $(1,k-1)$-weighted degree.

Following the ideas of Lee and O'Sullivan~\cite{leeOSullivan08}, we can first remark that $\Mod[\sell]{M}$ is an $\Fqx$-module.
Second, we can give an explicit basis for $\Mod[\sell]{M}$.
Define first two polynomials $G(X) \defeq \prod_{i=0}^{n-1}(X-\alpha_i)$ as well as $R(X)$ in $\Fqx$ as the unique Lagrange interpolation polynomial going through the points $(\alpha_i, r_i^\prime)$ for $i = 0, \ldots, n-1$.
Denote by $\yC Q t(X)$ the $Y^t$-coefficient of $Q(X,Y)$ when $Q$ is regarded over $\Fqx[Y]$.
\begin{lem}\label{lem_QdivG}
  Let $Q(X,Y) = \sum_{i=0}^t \yC Q i(X)Y^i \in \Mod[\sell]{M}$ with $\wdeg{0}{1}Q = t < s$. Then $G(X)^{s-t} \mid \yC Q t(X)$.
\end{lem}
\begin{proof}
  $Q(X,Y)$ interpolates the $n$ points $(\alpha_i, r_i^{\prime})$ with multiplicity $s$, so for any $i$, $Q(X+\alpha_i, Y+r_i^{\prime}) = \sum_{j=0}^t \yC Q j(X+\alpha_j)(Y+r_j^{\prime})^j$ has no monomials of total degree less than $s$.
  Multiplying out the $(Y+r_j^{\prime})^j$-terms, $\yC Q t(X+\alpha_j)Y^t$ is the only term with $Y$-degree $t$.
  Therefore $\yC Q t (X+\alpha_j)$ can have no monomials of degree less than $s-t$, which implies $(X-\alpha_i)^{s-t} \mid \yC Q t(X)$.
  As this holds for any $i$, we proved the lemma.
  \qed
\end{proof}

\begin{thm}
  \label{thm_Mbasis}
  The module $\Mod[\sell]{M}$ is generated as an $\Fqx$-module by the $\ell+1$ polynomials $P\T i(X,Y) \in \Fqxy$ given by
  \begin{align*}
    P\T t(X,Y) &= G(X)^{s-t}(Y-R(X))^t,         && \textrm{for } 0 \leq t < s,\\
    P\T t(X,Y) &= Y^{t-s}(Y-R(X))^s,            && \textrm{for } s \leq t \leq \ell.
  \end{align*}
\end{thm}
\begin{proof}
  It is easy to see that each $P\T t(X,Y) \in \Mod[\sell]{M}$ since both $G(X)$ and $(Y-R(X))$ go through the $n$ points $(\alpha_i, r_i^\prime)$ with multiplicity one, and that $G(X)$ and $(Y-R(X))$ divide $P\T t (X,Y)$ with total power $s$ for each $t$.

  To see that any element of $\Mod[\sell]{M}$ can be written as an $\Fqx$-combination of the $P\T t(X,Y)$, let $Q(X,Y)$ be some element of $\Mod[\sell]{M}$.
  Then the polynomial $Q\T{\ell-1}(X,Y) = Q(X,Y) - \yC Q \ell (X) P\T \ell(X,Y) $ has $Y$-degree at most $\ell-1$. Since both $Q(X,Y)$ and $P\T \ell(X,Y)$ are in $\Mod[\sell]{M}$, so must $Q\T{\ell-1}(X,Y)$ be in $\Mod[\sell]{M}$.
  Since $P\T t(X,Y)$ has $Y$-degree $t$ and $\yC{P\T t} t(X) = 1$ for $t = \ell, \ell-1, \ldots, s$, we can continue reducing this way until we reach a $Q\T{s-1}(X,Y) \in \Mod[\sell]{M}$ with $Y$-degree at most $s-1$.
  From then on, we have $\yC{P\T t} t(X) = G(X)^{s-t}$, but by Lemma~\ref{lem_QdivG}, we must also have $G(X) \mid \yC{Q\T{s-1}}{s-1}(X)$. Therefore, we can reduce by $P\T {s-1}(X,Y)$.
  This can be continued with the remaining $P\T t(X,Y)$, eventually reducing the remainder to 0.
  \qed
\end{proof}
We can represent the basis of $\Mod[\sell]{M}$ by the $(\ell+1) \times (\ell+1)$ matrix $\M[\sell]{A}=\| P^{(i)}_{[j]}(X)\|_{i=0,j=0}^{\ell, \ell}$ over $\Fqx$; more explicitly we have:
\begin{equation} \label{eq_BasisOfModul}
  \M[\sell]{A} \defeq
  \left(
    \begin{array}{ccccccc}
      G^s                                &                            &          &        &                              &                    & \\
      G^{s-1}(-R)                        & G^{s-1}                    &          &        & \multicolumn{2}{c}{\multirow{2}{*}{\mbox{\Large 0}}} & \\
G^{s-2}(-R)^2                        & 2 G^{s-2}(-R)                    & G^{s-2}         &        &  & \\
      
       \vdots                            &                            &   \ddots  &  &                               &                      \\
      (-R)^s                             & \binom{s}{1}(-R)^{s-1}     & \dots    & 1      &                               &                      \\
                                         & (-R)^{s}                   &          & \dots  & 1                             &                    & \\
        \multirow{2}{*}{\mbox{\Large 0}} & \multicolumn{2}{c}{\ddots} &          &        & \ddots                        &                      \\
                                         &                            & (-R)^{s} & \ldots &                               & & 1 
    \end{array}
  \right).
\end{equation}
Any $\Fqx$-linear combination of rows of $\M[\sell] A$ thus corresponds to an element in $\Mod[\sell] M$ by its $t$th position being the $\Fqx$-coefficient to $Y^t$.
All other bases of $\Mod[\sell]{M}$ can similarly be represented by matrices, and these will be unimodular equivalent to $\M[\sell]{A}$, i.e., they can be obtained by multiplying $\M[\sell] A$ on the left with an invertible matrix over $\Fqx$.

Extending the work of Lee and O'Sullivan~\cite{leeOSullivan08}, Beelen and Brander~\cite{beelenBrander} gave a fast algorithm for computing a satisfactory $Q(X,Y)$: start with $\M[\sell]{A}$
as a basis of $\Mod[\sell] M$ and compute a different, ``minimal'' basis of $\Mod[\sell]{M}$ where an element of minimal $(1, k-1)$-weighted degree appears directly.

In the following section, we give further details on how to compute such a basis, but our ultimate aim in Section \ref{sec_iterative} is different: we will use a minimal basis of $\Mod[\sell] M$ to efficiently compute one for $\Mod[\hat s,\hat \ell] M$ for $\hat s \geq s$ and $\hat \ell > \ell$.
This will allow an iterative refinement for increasing $s$ and $\ell$, where after each step we have such a minimal basis for $\Mod[\sell] M$. We then exploit this added flexibility in our multi-trial algorithm.

\section{Module Minimisation} \label{sec_modulemini}
Given a basis of $\Mod[\sell] M$, e.g. $\M[\sell] A$, the module minimisation here refers to the process of obtaining a new basis, which is the smallest among all bases of $\Mod[\sell] M$ in a precise sense.
We will define this and connect various known properties of such matrices.
We will then show how to perform this minimisation using the Mulders--Storjohann algorithm~\cite{mulders03}, reanalyse its performance and connect it to Alekhnovich's algorithm~\cite{alekhnovich05}.

\begin{defi}[Weak Popov Form \protect{\cite{mulders03}}]
  \label{def_weakpopov}
  A matrix $\M V$ over $\Fqx$ is in \emph{weak Popov form} if the leading position of each row is different.
\end{defi}
We are essentially interested in short vectors in a module, and the following lemma shows that the simple concept of weak Popov form will provide this.
It is a paraphrasing of \cite[Proposition 2.3]{alekhnovich05} and we omit the proof.
\begin{lem}[Minimal Degree]
  \label{lem_minrow}
  If a square matrix $\M V$ over $\Fqx$ is in weak Popov form, then one of its rows has minimal degree of all vectors in the row space of $\M V$.
\end{lem}
Denote now by $\Mod[\ell]{\M W}$ the diagonal $(\ell+1) \times (\ell+1)$ matrix over $\Fqx$:
\begin{equation} \label{eq_DiagWeight}
  \Mod[\ell]{\M W} \defeq \diag \left(1,X^{k-1}, \dots, X^{\ell(k-1)} \right).
\end{equation}
Since we seek a polynomial of minimal $(1,k-1)$-weighted degree, we also need the following corollary.
\begin{cor}[Minimal Weighted Degree] \label{cor_WeightSol}
  Let $\M B \in \Fqx^{(\ell+1) \times (\ell+1)}$ be the matrix representation of a basis of $\Mod[s,\ell] M$.
  If $\M B  \Mod[\ell]{\M W}$ is in weak Popov form, then one of the rows of $\M B$ corresponds to a polynomial in $\Mod[s,\ell] M$ with minimal $(1,k-1)$-weighted degree.
\end{cor}
\begin{proof}
  Let $\M{\tilde B} = \M B  \Mod[\ell]{\M W}$.
  Now, $\M{\tilde B}$ will correspond to the basis of an $\Fqx$-module $\tilde M$ isomorphic to $\Mod[s,\ell] M$, where an element $Q(X,Y) \in \Mod[s,\ell] M$ is mapped to $Q(X,X^{k-1}Y) \in \tilde M$.
  By Lemma \ref{lem_minrow}, the row of minimal degree in $\M{\tilde B}$ corresponds to an element of $\tilde M$ with minimal $X$-degree.
  Therefore, the same row of $\M B$ corresponds to an element of $\Mod[s,\ell] M$ with minimal $(1,k-1)$-weighted degree.
  \qed
\end{proof}
If for some matrix $\M B \in \Fqx^{(\ell+1) \times (\ell+1)}$, $\M B  \Mod[\ell]{\M W}$ is in weak Popov form, we say that $\M B$ is in \emph{weighted weak Popov form}.

We introduce what will turn out to be a measure of how far a matrix is from being in weak Popov form.
\begin{defi}[Orthogonality Defect \protect{\cite{lenstra85}}]
  \label{def_orthogonality}
The \emph{orthogonality defect} of a square matrix $\M V$ over $\Fqx$ is defined as
\begin{equation*}
\OD{\M V} \defeq \deg \M V - \deg \det \M V.
\end{equation*}
\end{defi}

\begin{lem}\label{lem_popovreduces}
  If a square matrix $\M V$ over $\Fqx$ is in weak Popov form, then $\OD{\M V} = 0$.
\end{lem}
\begin{proof}
  Let $\vec{v}_0,\ldots,\vec{v}_{m-1}$ be the rows of $\M{V} \in \Fqx^{m \times m}$ and $\vec v_i = (v_{i,0},\ldots, v_{i,m-1})$.
  In the alternating sum-expression for $\det \M V$, the term $\prod_{i=0}^{m-1} \LT{\vec{v}_i}$ will occur since the leading positions of $\vec{v}_i$ are all different.
  Thus $\deg \det \M V = \sum_{i=0}^{m-1} \deg \LT{\vec{v}_i} = \deg \M V$ unless leading term cancellation occurs in the determinant expression.
  However, no other term in the determinant has this degree: regard some (unsigned) term in $\det \M V$, say $t = \prod_{i=0}^{m-1} v_{i, {\sigma(i)}}$ for some permutation $\sigma \in S_m$.
  If not $\sigma(i) = \LP{\vec{v}_i}$ for all $i$ (as defined in~\eqref{def-leadingposition}), then there must be an $i$ such that $\sigma(i) > \LP{\vec{v}_i}$ since $\sum_j \sigma(j)$ is the same for all $\sigma \in S_m$.
  Thus, $\deg v_{i, {\sigma(i)}} < \deg v_{i, \LP{\vec{v}_i}}$.
  As none of the other terms in $t$ can have greater degree than their corresponding row's leading term, we get $\deg t < \sum_{i=0}^{m-1} \deg \LT{\vec{v}_i}$.
  Thus, $\OD{\M V} = 0$.
  \qed
\end{proof}

\begin{rem}
  The weak Popov form is highly related to minimal Gr\"obner bases of the row space module, using a term order where vectors in $\Fqx^m$ are ordered according to their degree; indeed the rows of a matrix in weak Popov form \emph{is} such a Gr\"obner basis (though the opposite is not always true).
  Similarly, the weighted weak Popov form has a corresponding weighted term order.
  In this light, Lemma \ref{lem_minrow} is simply the familiar assertion that a Gr\"obner basis over such a term order must contain a ``minimal'' element of the module.
  See e.g.~\cite[Chapter 2]{nielsenThesis} for more details on this correspondence.
  The language of Gr\"obner bases was employed in the related works of \cite{beelenBrander,leeOSullivan08}.
\end{rem}

\subsection{Algorithms}
\label{ssec_algorithms}

\begin{defi}[Row Reduction]
  \label{def_rowred} 
  Applying a \emph{row reduction} on a matrix over $\Fqx$ means to find two different rows $\vec{v}_i, \vec{v}_j$, $\deg \vec{v}_i < \deg \vec{v}_j$ and such that $\LP{\vec{v}_i} = \LP{\vec{v}_j}$, and then replacing $\vec{v}_j$ with $\vec{v}_j - \alpha X^\delta \vec{v}_i$ where $\alpha \in \Fq$ and $\delta \in \ZZ_+$ are chosen such that the leading term of the polynomial $\LT{\vec{v}_j}$ is cancelled.
\end{defi}
Algorithm~\ref{alg_mulders} is due to Mulders and Storjohann~\cite{mulders03}.
Our proof of it is similar, though we have related the termination condition to the orthogonality defect, restricting it to only square matrices.

\printalgos{
  \caption{Mulders-Storjohann~\cite{mulders03}}
  \label{alg_mulders}
   \KwIn{$\M V \in \Fqx^{m \times m}$}
   \KwOut{A matrix unimodular equivalent to $\M V$ and in weak Popov form.}
   \BlankLine
   Apply row reductions as in Definition~\ref{def_rowred} on the rows of $\M V$ until no longer possible \;
   \Return this matrix. \;
}

\def\vecVal{\psi}
Introduce for the proof a \emph{value function} $\vecVal: \Fqx^m \rightarrow \NN_0$ as $\vecVal(\vec v) = m\deg \vec v + \LP{\vec v}$.
First let us consider the following lemma.

\begin{lem}
  \label{lem_rowreddec}
  If we replace $\vec v_j$ with $\vec v_j'$ in a row reduction, then $\vecVal(\vec v_j') < \vecVal(\vec v_j)$.
\end{lem}
\begin{proof}
  We cannot have $\deg \vec v_j' > \deg \vec v_j$ since all terms of both $\vec v_j$ and $\alpha X^\delta \vec v_i$ have degree at most $\deg \vec v_j$.
  If $\deg \vec v_j' < \deg \vec v_j$ we are done since $\LP{\vec v_j'} < m$, so assume $\deg \vec v_j' = \deg \vec v_j$.
  Let $h = \LP{\vec v_j} = \LP{\vec v_i}$.
  By the definition of the leading position, all terms in both $\vec v_j$ and $\alpha X^\delta \vec v_i$ to the right of $h$ must have degree less than $\deg \vec v_j$, and so also all terms in $\vec v_j'$ to the right of $h$ satisfies this.
  The row reduction ensures that $\deg v'_{j,h} < \deg v_{j,h}$, so it must then be the case that $\LP{\vec v_j'} < h$.
\end{proof}

\begin{thm}
  \label{lem_mulders}
  Algorithm~\ref{alg_mulders} is correct. For a matrix $\M V \in \Fqx^{m \times m}$, it performs fewer than $m(\OD{\M V}+(m+1)/2)$ row reductions and has asymptotic complexity $O(m^2\OD{\M V}N)$ where $N$ is the maximal degree of any term in $\M V$.
\end{thm}
\begin{proof}
  If Algorithm~\ref{alg_mulders} terminates, the output matrix must be unimodular equivalent to the input since it is reached by a finite number of row-operations.
  Since we can apply row reductions on a matrix if and only if it is not in weak Popov form, Algorithm~\ref{alg_mulders} must bring $\M V$ to this form.

  Termination follows directly from Lemma \ref{lem_rowreddec} since the value of a row decreases each time a row reduction is performed.
  To be more precise, we furthermore see that the maximal number of row reductions performed on $\M V$ before reaching a matrix $\M U$ in weak Popov form is at most $\sum_{i=0}^{m-1} \vecVal(\vec v_i) - \vecVal(\vec u_i)$. Expanding this, we get
  \begin{align*}
    \sum_{i=0}^{m-1} \vecVal(\vec v_i) - \vecVal(\vec u_i)
       &= \sum_{i=0}^{m-1} \big( m(\deg \vec v_i - \deg \vec u_i) + \LP{\vec v_i} - \LP{\vec u_i} \big)
    \\ &\textstyle = m(\deg \M V - \deg \M U) + \sum_{i=0}^{m-1} \LP{\vec v_i} - \binom m 2
    \\ &< m \Big( \OD{\M V} + \tfrac{m+1}{2} \Big)
  \end{align*}
  where we use $\deg \M U = \deg \det \M U = \deg \det \M V$ and that the $\LP{\vec u_i}$ are all different.

  For the asymptotic complexity, note that during the algorithm, no polynomial in the matrix will have larger degree than $N$.
  The estimate is reached simply by remarking that one row reduction consists of $m$ times scaling and adding two such polynomials.
  \qed
\end{proof}

Let us consider an example to illustrate all the above.
\begin{exa}[Orthogonality Defect and Weak-Popov Form] \label{ex_weakpopov}
Let us consider the following matrices $\M{V}_0, \dots, \M{V}_3 \in \Fx{2}^{3 \times 3}$. From matrix $\M{V}_i$ to $\M{V}_{i+1}$ we performed one row-operation:
\begin{tabular}{lllllll}
$\M{V}_0 =$ & 
&
$\M{V}_1 =$ &
&
$\M{V}_2 =$ &
&
$\M{V}_3 =$ \\
$\left(\begin{array}{ccc}
        1 & X^2 & X \\
        0 & X^3 & X^2 \\
        X & 1 & 0 
    \end{array}\right)$ &
\hspace{-.4cm} $\xrightarrow{(0,1)} $ \hspace{-.4cm} &
$ \left(\begin{array}{ccc}
        1 & X^2 & X \\
        X & 0 & 0 \\
        X & 1 & 0 
    \end{array}\right)$ &
\hspace{-.4cm} $\xrightarrow{(2,1)}$ \hspace{-.4cm} &
$  \left(\begin{array}{ccc}
        1 & X^2 & X \\
        X & 0 & 0 \\
        0 & 1 & 0 
    \end{array}\right)$ &
\hspace{-.4cm} $\xrightarrow{(0,2)}$ \hspace{-.4cm}  &
$  \left(\begin{array}{ccc}
        1 & 0 & X \\
        X & 0 & 0 \\
        0 & 1 & 0 
    \end{array}\right)$,
\end{tabular}
where the indexes $(i_1,i_2)$ on the arrow indicated the concerned rows.
The orthogonality defect is decreasing; $\OD{\M{V}_0} = 3 \rightarrow \OD{\M{V}_1} = 2 \rightarrow \OD{\M{V}_2} = 1 \rightarrow \OD{\M{V}_3} = 0$, and $\M V_3$ is in weak Popov form.
\end{exa}

In \cite{alekhnovich05}, Alekhnovich gave a divide-\&-conquer variant of the Mulders--Storjohann-algorithm: the same row reductions are performed but structured in a binary computation tree, where work is done on matrices of progressively smaller degree towards the bottom, ending with essentially $\Fq$-matrices at the leaves.
Alekhnovich does not seem to have been aware of the work of Mulders and Storjohann, and basically reinvented their algorithm before giving his divide-\&-conquer variant.

For square matrices, we can improve upon the complexity analysis that Alekhnovich gave by using the concept of orthogonality defect; this will be crucial for our aims.
\begin{lem}[Alekhnovich's Algorithm~\cite{alekhnovich05}] \label{lem_alekcompl}
Alekhnovich's algorithm inputs a matrix $\M V \in \Fqx^{m\times m}$ and outputs a unimodular equivalent matrix which is in weak Popov form.
  Let $N$ be the greatest degree of a term in $\M V$.
  If $N \in O(\OD{\M V})$ then the algorithm has asymptotic complexity:
  \begin{equation*}
  O\big( m^3 \OD{\M V}\log^2 \OD{\M V} \log\log \OD{\M V} \big) \quad \text{operations over $\Fq$}.
  \end{equation*}
\end{lem}
\begin{proof}
  The description of the algorithm as well as the proof of its correctness can be found in~\cite{alekhnovich05}.
  We only prove the claim on the complexity.
  The method $R(\M V, t)$ of \cite{alekhnovich05} computes a unimodular matrix $\M U$ such that $\deg(\M U\M V) \leq \deg \M V - t$ or $\M U\M V$ is in weak Popov form.
  According to \cite[Lemma 2.10]{alekhnovich05}, the asymptotic complexity of this computation is  in $O(m^3 t \log^2 t \log\log t)$.
  Due to Lemma~\ref{lem_popovreduces}, we can set $t = \OD{\M V}$ to be sure that $\M U\M V$ is in weak Popov form.
  What remains is just to compute the product $\M U\M V$.
  Due to \cite[Lemma 2.8]{alekhnovich05}, each entry in $\M U$ can be represented as $p(X)X^d$ for some $d \in \NN_0$ and $p(X) \in \Fqx$ of degree at most $2t$.
  If therefore $N \in O(\OD{\M V})$, the complexity of performing the matrix multiplication using the naive algorithm is $O(m^3\OD{\M V})$.
  \qed
\end{proof}
The Beelen--Brander interpolation algorithm~\cite{beelenBrander} works simply by computing $\M[\sell] A$ and then applying Alekhnovich's algorithm on $\M[\sell] A \M[\ell] W$; a minimal $(1,k-1)$-weighted polynomial in $\Mod[\sell] M$ can then be directly retrieved as a row in the reduced matrix, after removing $\M[\ell] W$.
The algorithm's complexity therefore follows from the above theorem once we have computed the orthogonality defect of $\M[\sell] A \M[\ell] W$:
\begin{lem}
  \label{lem_ODofA}
$\OD{\M[\sell] A \M[\ell] W} = \half(2\ell - s + 1)s(\deg R - k+1) \ < \  \ell s (n-k)$.
\end{lem}
\begin{proof}
  We will compute first $\deg (\M[\sell] A \M[\ell] W)$ and then $\deg \det (\M[\sell]A \M[\ell] W)$.

For the former, we have $\deg(\M[\sell] A \M[\ell] W) = \sum_{t=0}^\ell \wdeg{1}{k-1} P^{(t)}(X,Y)$, where the $P^{(t)}(X,Y)$ are as in Theorem \ref{thm_Mbasis}.
  Note that whenever $\vec r$ is not a codeword then $\deg R \geq k$.
  Therefore, $\wdeg{1}{k-1} (Y-R(X))^t = t\deg R(X)$ and so
  \begin{align*}
  \wdeg{1}{k-1} P^{(t)}(X,Y) & = (s-t)n + t \deg R(X)    & \textrm{for } t < s, \\
  \wdeg{1}{k-1} P^{(t)}(X,Y) & = (t-s)(k-1) + s\deg R(X) & \textrm{for } t \geq s.
  \end{align*}
  This gives
  \[
  \deg(\M[\sell] A \M[\ell] W) = \tbinom{s+1} 2 n + \tbinom{\ell-s+1} 2 (k-1)
  + \big(\tbinom{s+1} 2 + (\ell-s)s\big)\deg R(X).
  \]
  Since $\M[\sell] A$ is lower triangular, the determinant is:
  \begin{align*}
  \det(\M[\sell] A \M[\ell] W) &= \prod_{t=0}^s G^{s-t} \prod_{t=0}^\ell X^{t(k-1)}, & \textnormal{and so} \\
  \deg \det(\M[\sell] A \M[\ell] W) &= \tbinom{s+1} 2 n + \tbinom{\ell+1} 2(k-1).
  \end{align*}
  The orthogonality defect can then be simplified to
  \begin{align*}
  \OD{\M[\sell] A \M[\ell] W}
  &= \tbinom{\ell-s+1} 2 (k-1) + \big(\tbinom{s+1} 2 + (\ell-s)s\big)\deg R(X) -  \tbinom{\ell+1} 2(k-1) \\
  &= \deg R(X)\big(s\ell - \half s^2 + \half s\big) - \half(k-1)\big(\ell^2 + \ell - (\ell-s+1)(\ell-s)\big) \\
  &= \half(2\ell - s + 1)s(\deg R(X) - k+1).
  \end{align*}
\vspace{-.4cm}
  \qed
\end{proof}
\begin{rem}
  In the earlier work of Lee and O'Sullivan \cite{leeOSullivan08}, they construct basically the same matrix as $\M[\sell] A$, but apply their own Gr\"obner basis algorithm on this.
  They also does not seem to have been aware of the work of Mulders and Storjohann~\cite{mulders03}, but their algorithm is basically a variant of Algorithm \ref{alg_mulders} which keeps the rows in a specific order.
\end{rem}

\section{Multi-Trial List Decoding}
\label{sec_iterative}

\subsection{Basic Idea}
\label{ssec_outline}

Using the results of the preceding section, we show in Section~\ref{ssec_steptypeone} that, given a basis of $\Mod[\sell] M$ as a matrix $\M[\sell] B$ in weighted weak Popov form, then we can write down a matrix $\M[s,\ell+1] C\stepI$ which is a basis of $\Mod[s,\ell+1] M$ and where $\OD{\M[s,\ell+1] C\stepI \M[\ell] W}$ is much lower than $\OD{\M[s,\ell+1] A \M[\ell] W}$.
This means that reducing $\M[s,\ell+1] C\stepI$ to weighted weak Popov form using Alekhnovich's algorithm~\cite{alekhnovich05} is faster than reducing $\M[s,\ell+1] A$.
We call this kind of refinement a ``micro-step of type I''.
In Section~\ref{ssec_steptypetwo}, we similarly give a way to refine a basis of $\Mod[\sell] M$ to one of $\Mod[s+1,\ell+1] M$, and we call this a ``micro-step of type II''.

If we first compute a basis in weighted weak Popov form of $\Mod[1,1] M$ using $\M[1,1] A$, we can perform a sequence of micro-steps of type I and II to compute a basis in  weighted weak Popov form of $\Mod[s,\ell] M$ for any $s,\ell$ with $\ell \geq s$.
After any step, having some intermediate $\is \leq s$, $\iell \leq \ell$, we will thus have a basis of $\Mod[\isell] M$ in weighted weak Popov form.
By Corollary~\ref{cor_WeightSol}, we could extract from $\M[\isell] B$ a $\iQ(X,Y) \in \Mod[\isell] M$ with minimal $(1,k-1)$-weighted degree.
Since it must satisfy the interpolation conditions of Theorem~\ref{thm_GSproblem}, and since the weighted degree is minimal among such polynomials, it must also satisfy the degree constraints for $\itau = \tau(\isell)$.
By that theorem any codeword with distance at most $\itau$ from $\vec r$ would then be represented by a root of $\iQ(X,Y)$.

Algorithm~\ref{alg_multitrial} is a generalisation and formalisation of this method.
For a given $\RS n k$ code, one chooses ultimate parameters $(s, \ell, \tau)$ being a permissible triple with $s \leq \ell$.
One also chooses a list of micro-steps and chooses after which micro-steps to attempt decoding; these choices are represented by a list $\var{C}$ consisting of $\var S_1$, $\var S_2$ and $\var{Root}$ elements.
This list must contain exactly $\ell-s$ $\var S_1$-elements and $s-1$ $\var S_2$-elements, as it begins by computing a basis for $\Mod[1,1] M$ and will end with a basis for $\Mod[\sell] M$.
Whenever there is a $\var{Root}$ element in the list, the algorithm performs root-finding and finds all codewords with distance at most $\itau = \tau(\isell)$ from $\vec r$; if this list is non-empty, the computation breaks and the list is returned.

The algorithm calls sub-functions which we explain informally: $\var{MicroStep1}$ and $\var{MicroStep2}$ will take $\isell$ and a basis in weighted weak Popov form for $\Mod[\isell] M$ and return a basis in weighted weak Popov form for $\Mod[\is,\iell+1] M$ respectively $\Mod[\is+1,\iell+1] M$; more detailed descriptions for these are given in Subsections~\ref{ssec_steptypeone} and \ref{ssec_steptypetwo}.
$\var{MinimalWeightedRow}$ finds a polynomial of minimal $(1,k-1)$-weighted degree in $\Mod[\isell]{M}$ given a basis in weighted weak Popov form (Corollary~\ref{cor_WeightSol}).
Finally, $\var{RootFinding}(Q, \tau)$ returns all $Y$-roots of $Q(X,Y)$ of degree less than $k$ and whose corresponding codeword has distance at most $\tau$ from the received word $\vec r$.

\printalgos{\caption{Multi-Trial Guruswami--Sudan Decoding}
\label{alg_multitrial} 
\SetKwInput{KwPre}{{Preprocessing}}
\SetKwInput{KwIn}{{Input}}
\SetKwInput{KwOut}{{Output}}
\BlankLine
\KwIn{\\
A $\RS{n}{k}$ code over $\Fq$ with $w_0,\dots,w_{n-1} \in \Fq^*$\;
The received vector $\vec r = (r_0, \dots, r_{n-1}) \in \Fqn{q}{n}$\;
A permissible triple $(s,\ell,\tau) \in \NN^3$\;
A list \var{C} with elements in $\{\var S_1, \var S_2, \var{Root} \}$ with $\ell-s$ instances of $\var S_1$ and $s-1$  instances of $\var S_2$\;
}
\BlankLine
\KwPre{\\
Calculate $r_i^{\prime} = r_i/w_i$ for all $i=0,\dots,n-1$\;
Construct $\M[1,1]{A}$, and compute $\M[1,1]{B}$ from $\M[1,1]{A} \M[1]{W}$ using Alekhnovich's algorithm\;
Initial parameters $(\is, \iell) \leftarrow (1, 1)$ \;}
\BlankLine
\For{each $\var c$ in $\var{C}$}{
	\If{$\var c = \var S_1$ \label{alg_line_firstif}}{
          \makebox[3.7em][l]{$\M[\is,\iell+1]{B}$} $\leftarrow \var{MicroStep1}(\is, \iell, \M[\isell]{B})$\;
          \makebox[3.7em][l]{$(\isell)$} $\leftarrow (\is,\iell+1)$\;
	}
	\If{$\var c = \var S_2$ \label{alg_line_secondif}} {
          \makebox[3.7em][l]{$\M[\is+1,\iell+1]{B}$} $\leftarrow \var{MicroStep2}(\is, \iell, \M[\isell]{B})$\;
          \makebox[3.7em][l]{$(\isell)$} $\leftarrow (\is+1,\iell+1)$\;
	}
	\If{$\var c = \var{Root}$ \label{alg_line_thirdif}} {
          \makebox[3.7em][l]{$Q(X,Y)$}  $\leftarrow \var{MinimalWeightedRow}(\M[\isell]{B})$ \label{alg_line_extractpoly}\;
          \If{$\var{RootFinding}(Q(X,Y), \tau(\isell)) \neq \emptyset$}{
            \Return this list\;
          }
        }
}
}
The correctness of Algorithm~\ref{alg_multitrial} for any possible choice of $s, \ell$ and $\var C$ follows from our discussion as well as Sections \ref{ssec_steptypeone} and \ref{ssec_steptypetwo}.
Before going to these two technical sections, we will discuss what possibilities the micro-steps of type I and II offer, and in particular, do not, with regards to decoding radii.

In the following two subsections we explain the details of the micro-steps.
In Section \ref{ssec_complanalysis}, we discuss the complexity of the method and how the choice of $\var C$ influences this.

\subsection{The Possible Refinement Paths}
The choice of $\var C$ provides much flexibility to the algorithm.
The two extreme cases are perhaps the most generally interesting: the one without any $\var{Root}$ elements except at the end, i.e.,~usual list-decoding; and the one with a $\var{Root}$ element each time the intermediate decoding radius $\hat\tau$ has increased, i.e.,~a variant of maximum-likelihood decoding up to a certain radius.

In Section \ref{ssec_complanalysis}, we discuss complexity concerns with regards to the chosen path; it turns out that the price of either type of micro-step is very comparable, and the worst-case complexity is completely unchanged by the choice of $\var C$.
However, in the case where we have multiple $\var{Root}$ elements we want to minimise the \emph{average} computation cost: considering that few errors occur much more frequently than many, we should therefore seek to reach each intermediate decoding radius after as few micro-steps as possible.

Since we do not have a refinement which increases only $s$, we are inherently limited in the possible paths we can choose, so the question arises if this limitation conflicts with our interest as given above.

First --- and most important --- for any given final decoding radius $\tau$, we mentioned in Section \ref{ssec_ipdecoding} that the corresponding parameters satisfy $s < \ell$, and so we can reach these values using only micro-steps of type I and II.

For the intermediate steps, the strongest condition we would like to have satisfied is the following:
Let $d/2 \leq \tau_1 < \ldots < \tau_m = \tau$ be the series of intermediate decoding radii where we would like to attempt decoding.
Let $(s_i,\ell_i)$ be chosen such that $(s_i,\ell_i,\tau_i)$ is permissible and either $s_i$ or $\ell_i$ is minimal possible for the given $\tau_i$.
Can then the sequence of parameters $(s_i,\ell_i)$ be reached by a series of micro-steps of type I and II?

Unfortunately, we do not have a formal proof of this statement.
However, we have verified for a large number of parameters that it is true.

\subsection{Micro-Step Type I: $(s,\ell) \mapsto (s, \ell+1)$} \label{ssec_steptypeone}
The function \var{MicroStep1} is based on the following lemma:
\begin{lem} \label{lem_MicroStepOne}
  If $B\T 0(X,Y), \ldots, B\T \ell(X,Y) \in \Fqxy$ is a basis of $\Mod[\sell]{M}$, then the following is a basis of $\Mod[\sellI] M$:
  \[
  B\T 0(X,Y),\ \ldots\ ,\ B\T \ell(X,Y), Y^{\ell-s+1}(Y-R(X))^s.
  \]
\end{lem}
\begin{proof}
  In the basis of $\Mod[\sellI]{M}$ given in Theorem~\ref{thm_Mbasis}, the first $\ell+1$ generators are the generators of $\Mod[\sell]{M}$.
  Thus, all of these can be described by any basis of $\Mod[\sellI]{M}$.
  The last remaining generator is exactly $Y^{\ell-s+1}(Y-R(X))^s$.
  \qed
\end{proof}
In particular, the above lemma holds for a basis of $\Mod[\sellI] M$ in weighted weak Popov form, represented by a matrix $\M[\sell]{B}$.
The following matrix thus represents a basis of $\Mod[\sellI]{M}$:
\begin{equation}\label{eqn_stepI}
  \M[\sellI] C\stepI=
  \left(\begin{array}{r}
      \begin{array}{@{}c|c@{}}
      \\[-0.2cm]
      \makebox[4cm][c]{$\M[\sell]{B} $}
      & 
      \makebox[1.5em][r]{$\vec 0^T$}
      \\[-0.2cm] \
    \end{array}
    \\ \hline \\[-.3cm]
    \begin{matrix}
      0 & \ldots & 0 & (-R)^s & \binom s 1 (-R)^{s-1} & \ldots & 1
    \end{matrix}
  \end{array}\right).
\end{equation}
\begin{lem}
  \label{lem_orthoI} $\OD{\M[\sellI]{C}\stepI \M[\ell+1]{W}} = s(\deg R - k + 1) \leq s(n-k)$.
\end{lem}
\begin{proof}
  We calculate the two quantities $\det(\M[\sellI] C\stepI \M[\ell+1]{W})$ and $\deg(\M[\sellI] C\stepI \M[\ell+1]{W})$.
  It is easy to see that
  \[
  \det(\M[\sellI] C\stepI \M[\ell+1]{W})
  = \det \M[\sell]{B} \det \M[\ell+1]{W}
  = \det \M[\sell]{B} \det \M[\ell]{W} X^{(\ell+1)(k-1)}.
  \]
  For the row-degree, this is clearly $\deg(\M[\sell]B\M[\ell] W)$ plus the row-degree of the last row.
  If and only if the received word is not a codeword then $\deg R \geq k$, so the leading term of the last row must be $(-R)^sX^{(\ell+1-s)(k-1)}$.
  Thus, we get
  \begin{align*}
    \OD{\M[\sellI]C\stepI \M[\ell+1] W}
    &= \big( \deg(\M[\sell] B \M[\ell] W) + s \deg R + (\ell+1-s)(k-1) \big)
    \\ &\qquad - \big( \deg\det(\M[\sell] B \M[\ell] W) + (\ell+1)(k-1) \big)
    \\ &= s(\deg R - k + 1),
  \end{align*}
  where the last step follows from Lemma~\ref{lem_popovreduces} as $\M[\sell]B \M[\ell] W$ is in weak Popov form.
  \qed
\end{proof}
\begin{cor}
  \label{cor_ComplMSI}
  The complexity of $\var{MicroStep1}(s,\ell,\M[\sell]B)$ is $O(\ell^3 s n \log^2 n \log\log n)$.
\end{cor}
\begin{proof}
  Follows by Lemma~\ref{lem_alekcompl}. Since $s \in O(n^2)$ we can leave out the $s$ in $\log$-terms.
  \qed
\end{proof}

\subsection{Micro-Step Type II: $(s,\ell) \mapsto (s+1, \ell+1)$} \label{ssec_steptypetwo}
The function \var{MicroStep2} is based on the following lemma:
\begin{lem} \label{lem_MicroStepTwo}
  If $B\T 0(X,Y), \ldots, B\T \ell(X,Y) \in \Fqxy$ is a basis of $\Mod[\sell]{M}$, then the
  following is a basis of $\Mod[\sellII]M$:
  \[
    G^{s+1}(X),\ B\T 0(X,Y)(Y-R(X)),\ \ldots\ ,\ B\T \ell(X,Y)(Y-R(X)).
  \]
\end{lem}
\begin{proof}
  Denote by $P_{\sell}\T 0(X,Y),\ldots,P_{\sell}\T \ell(X,Y)$ the basis of $\Mod[\sell]M$ as given in Theorem~\ref{thm_Mbasis}, and by $P_{\sellII}\T 0(X,Y),\ldots,P_{\sellII}\T{\ell+1}(X,Y)$ the basis of $\Mod[\sellII] M$.
  Then observe that for $t > 0$, we have $P_{\sellII}\T t = P_{\sell}\T{t-1}(Y-R(X))$.
  Since the $B\T t(X,Y)$ form a basis of $\Mod[\sell]{M}$, each $P_{\sell}\T t$ is expressible as an $\Fqx$-combination of these, and thus for $t>0$, $P_{\sellII}\T t$ is expressible as an $\Fqx$-combination of the $B\T t(X,Y)(Y-R(X))$.
  Remaining is then only $P_{\sellII}\T 0(X,Y) = G^{s+1}(X)$.
  \qed
\end{proof}
As before, we can use the above with the basis $\M[\sell] B$ of $\Mod[\sell] M$ in weighted weak Popov form, found in the previous iteration of our algorithm.
Recall that multiplying by $Y$ translates in the matrix representation to shifting one column to the right, so the following matrix represents a basis of $\Mod[\sellII]{M}$:
\begin{equation}\label{eqn_stepII}
  \M[\sellII]C\stepII =
    \left(\begin{array}{@{}c|c@{}}
        G^{s+1} & \vec 0 \\\hline\\[-.2cm]
        \vec 0^T & \makebox[1cm][c]{$\vec 0$}
          \\[0.1cm]
    \end{array}\right)
    +
    \left(\begin{array}{@{}c|c@{}}
        0 &  \vec 0 \\\hline\\[-.2cm]
        \vec 0^T & \makebox[1cm][c]{$\M[\sell]{B}$}
          \\[0.1cm]
    \end{array}\right)
    - R\cdot
    \left(\begin{array}{@{}c|c@{}}
        \vec 0 & 0 \\\hline\\[-.2cm]
        \makebox[1cm][c]{$\M[\sell]{B}$} & \vec 0^T
          \\[0.1cm]
    \end{array}\right).
\end{equation}
\begin{lem}
  \label{lem:Cortho}
  $\OD{\M[\sellII]C\stepII \M[\ell+1] W}
  = (\ell+1)(\deg R - k + 1)
  \leq (\ell+1)(n-k)$.
\end{lem}
\begin{proof}
  We compute $\deg(\M[\sellII]C\stepII \M[\ell+1] W)$ and $\deg \det(\M[\sellII]C\stepII \M[\ell+1] W)$.
  For the former, obviously the first row has degree $(s+1)n$.
  Let $\vec{b}_i$ denote the $i$th row of $\M[\sell] B$ and $\vec{b}'_i$ denote the $i$th row of $\M[\sell] B \M[\ell] W$.
  The $(i+1)$th row of $\M[\sellII]C\stepII \M[\ell+1]W$ has the form
  \[
    \big[(0 \mid \vec{b}_i) - R(\vec{b}_i \mid 0)\big]\M[\ell+1] W
      = (0 \mid \vec{b}'_i)X^{k-1} - R(\vec{b}'_i \mid 0).
  \]
  If and only if the received word is not a codeword, then $\deg R \geq k$.
  In this case, the leading term of $R \vec{b}'_i$ must have greater degree than any term in $X^{k-1}\vec{b}'_i$.
  Thus the degree of the above row is $\deg R + \deg \vec{b}'_i$.
  Summing up we get
  \begin{align*}
    \deg \M[\sellII]C\stepII \M[\ell+1] W &= (s+1)n + \sum_{i=0}^\ell \left(\deg R + \deg \vec{b}'_i \right) \\
    &= (s+1)n + (\ell+1)\deg R + \deg(\M[\sell]{B} \M[\ell] W).
  \end{align*}
  For the determinant, observe that
  \begin{align*}
    \det(\M[\sellII]C\stepII \M[\ell+1] W)
    &= \det(\M[\sellII]C\stepII) \det(\M[\ell+1] W) \\
    &= G^{s+1} \det \tilde{\M B} \det \M[\ell] W \cdot X^{(\ell+1)(k-1)},
  \end{align*}
  where $\tilde{\M B} = \M[\sell]{B} - R \left( {\M[\sell] {\grave B}} \ \big|\  \vec 0^T \right)$ and $\M[\sell]{\grave B}$ is all but the zeroth column of $\M[\sell] B$.
  This means $\tilde{\M B}$ can be obtained by starting from $\M[\sell] B$ and iteratively adding the $(j+1)$th column of $\M[\sell] B$ scaled by $R(X)$ to the $j$th column, with $j$ starting from $0$ up to $\ell-1$.
  Since each of these will add a scaled version of an existing column in the matrix, this does not change the determinant.
  Thus, $\det \tilde{\M B} = \det \M[\sell] B$.
  But then $\det \tilde{\M B} \det \M[\ell] W = \det(\M[\sell] B \M[\ell] W)$ and so $\deg(\det \tilde{\M B}\det \M[\ell] W) = \deg(\M[\sell] B \M[\ell] W)$ by Lemma~\ref{lem_popovreduces} since $\M[\sell] B \M[\ell] W$ is in weak Popov form.
  Thus we get
  \[
    \deg \det(\M[\sellII]C\stepII \M[\ell+1] W) =
      (s+1)n + \deg(\M[\sell] B \M[\ell] W) + (\ell+1)(k-1).
  \]
  The lemma follows from the difference of the two calculated quantities.
  \qed
\end{proof}
\begin{cor}
  \label{cor_ComplMSII}
  The complexity of $\var{MicroStep2}(s,\ell,\M[\sell]B)$ is $O(\ell^4 n \log^2 n \log\log n)$.
\end{cor}

\begin{exa} \label{ex_algo164}
We consider again the \RS{16}{4} code over $\F{17}$ of Example~\ref{ex_rs164param}, and specify now that $\alpha_i = i+1$ and $w_i = 1$ for $i=0,\ldots,15$.
The aimed decoding radius is $\tau = 8$ and therefore the permissible triple $(s,\ell,\tau) = (2,4,8)$ should be reached iteratively by Algorithm~\ref{alg_multitrial}.
To maximise the decoding radius during the procedure, we could choose the following
sequence of intermediate parameters $(\hat s,\hat \ell,\hat \tau)$:
\begin{equation*}
(1,1,6)
\xrightarrow{\textup I}
(1,2,7)
\xrightarrow{\textup{II}}
(2,3,7)
\xrightarrow{\textup I}
(2,4,8).
\end{equation*}
We perform root-finding only if the decoding radius is increased. Therefore, the list $\var{C}$ of operations becomes:
\begin{equation*}
\var{C} = \big\{ \var{Root}, \var S_1, \var{Root}, \var S_2, \var S_1, \var{Root} \big\}.
\end{equation*}
With the information polynomial $f(X) = 2X^2+10X+6$, we obtain with~\eqref{eq_defRScode} the following codeword:
\begin{equation*}
\vec{c} = (1, 0, 3, 10, 4, 2, 4, 10, 3, 0, 1, 6, 15, 11, 11, 15).
\end{equation*}
Consider that $\vec{r} = (1, 15, 12, 13, 4, 7, 4, 10, 1, 0, 1, 10, 2, 11, 11, 10)$ was received, i.e.,~that the error $\vec{e} = (0, 15, 9, 3, 0, 5, 0, 0, 15, 0, 0, 4, 4, 0, 0, 12)$ of weight 8 occurred.

We will depict the degrees of the polynomials in the matrices in the iterative decoding process.
These are for this particular received word, but for a generic received word, the degrees are the same.
For some $p(X) \in \Fqx$, we will write $p(X) \preceq t$ for $t \in \NN_0$ if $\deg p(X) = t$, and $p(X) \preceq \bot$ if $p(X) = 0$, and we extend $\preceq$ element-wise to matrices.
To begin with, we have:
\begin{align*}
\M[1,1]{A} = \left(\begin{array}{rr}
                G & 0 \\
                -R & 1
              \end{array}\right) \quad
\M[1,1]{A} \preceq \left(\begin{array}{rr}
      16 & \bot \\
      15 & 0
  \end{array}\right)  \quad
\M[1,1]{A} \M[1]{W} \preceq \left(\begin{array}{rr}
      16 & \bot \\
      15 & 3
  \end{array}\right)
\end{align*}
according to~\eqref{eq_BasisOfModul} and~\eqref{eq_DiagWeight}.
We then apply Alekhnovich's algorithm on $\M[1,1]{A} \M[1]W$ to obtain $\M[1,1]{B} \M[1]{W}$ which is in weak Popov form. From this we can easily scale down the columns again to obtain $\M[1,1]{B}$.
It took 11 row reductions, while $(\hat \ell+1)(\OD{\M[1,1]{A}\M[1] W}+\hat \ell+1) = 28$ was the upper bound, according to Lemma \ref{lem_mulders}.
We obtain
\begin{align*}
  \M[1,1]{B} \preceq \left(\begin{array}{rr}
    10 & 6 \\
    9 & 6
  \end{array}\right) \quad \quad
  \M[1,1]{B}\M[1]{W} \preceq \left(\begin{array}{rr}
    10 & 9 \\
    9 & 9
  \end{array}\right).
\end{align*}
The first element in $\var C$ is $\var{Root}$, so we pick the second row of $\M[1,1]{B}$, since it has weighted degree less than $10$, and we interpret it as a polynomial:
\begin{align*}
Q_{1,1}(X,Y) & = (14X^6 + 9X^4 + 9X^3 + 14X^2 + 4X + 1)Y \\
& \quad + 13X^9 + 10X^8 + 7X^6 + 16X^5 + 8X^3 + 12X^2 + 3X + 16.
\end{align*} 
Root-finding of $Q_{1,1}(X,Y)$ yields no results.
The next element in $\var C$ is $\var S_1$, so we move to the next intermediate parameters $(\hat s,\hat \ell,\hat \tau) = (1,2,7)$.
From~\eqref{eqn_stepI}, we get
\begin{align*}
  \M[1,2]C\stepI  =
    \left(\begin{array}{r}
            \begin{array}{@{}c|c@{}}
        \makebox[1cm][c]{$\M[1,1]{B}$}
        & 
        \begin{matrix} 0 \\ 0 \end{matrix}
      \end{array}
      \\ \hline \\[-.4cm]
      \begin{matrix}
        0 & -R & 1
      \end{matrix}
    \end{array}\right)
  \qquad
  \textnormal{\it and so}
  \qquad
      \M[1,2]C\stepI \preceq \left(\begin{matrix}
      10 & 6 & \bot \\
      9 & 6 & \bot \\
      \bot & 15 & 0
    \end{matrix}\right) \quad \quad
      \M[1,2]C\stepI\M[2]{W} \preceq \left(\begin{matrix}
      10 & 9 & \bot \\
      9 & 9 & \bot \\
      \bot & 18 & 6
    \end{matrix}\right).
\end{align*}
Running Alekhnovich's algorithm on $\M[1,2]C\stepI\M[2]{W}$, we obtain:
\begin{align*}
    \M[1,2]{B} \preceq \left(\begin{matrix}
    9 & 4 & 2 \\
    8 & 5 & 2 \\
    8 & 5 & 1
  \end{matrix}\right) \quad \quad
    \M[1,2]{B} \M[2]{W} \preceq \left(\begin{matrix}
    9 & 7 & 8 \\
    8 & 8 & 8 \\
    8 & 8 & 7
  \end{matrix}\right).
\end{align*}
Since $\OD{\M[1,2]{B} \M[2]{W}} = 12$, Lemma~\ref{lem_mulders} gives $45$ as the upper bound on the number of row reductions, but it was done with only 24.

In the next iteration we again meet a $\var{Root}$.
For our polynomial we can pick either the second or third row of $\M[1,2] B$ since both have weighted degree $8 < \hat s(n-\hat \tau) = 9$; we choose the second and obtain:
\begin{align*}
Q_{1,2}(X,Y) & = (15X^2 + 8X)Y^2 + (5X^5 + 2X^4 + 2X^3 + 10X^2 + 5X + 1)Y \\
& \quad + 14X^8 + 16X^7 + 7X^6 + 8X^5 + 9X^4 + 9X^3 + X^2 + 9X + 15.
\end{align*}
Again root-finding yields no results.
The next element in $\var C$ is $\var S_2$ and we get intermediate parameters $(\hat s,\hat \ell, \hat \tau) = (2,3,7)$.
We construct $\M[2,3] C\stepII$ according to \eqref{eqn_stepII} which gives:
\begin{align*}
  \M[2,3]C\stepII \preceq \left(\begin{matrix}
      32 & \bot & \bot & \bot \\
      24 & 19 & 17 & 2 \\
      23 & 20 & 17 & 2 \\
      23 & 20 & 16 & 1
  \end{matrix}\right) \quad \quad
  \M[2,3]C\stepII \M[3]{W} \preceq \left(\begin{matrix}
      32 & \bot & \bot & \bot \\
      24 & 22 & 23 & 11 \\
      23 & 23 & 23 & 11 \\
      23 & 23 & 22 & 10
  \end{matrix}\right).
\end{align*}
We needed 90 row reductions to reduce $\M[2,3]C\stepII \M[3]{W}$ to weak Popov form, while the upper bound is 160, since we calculated $\OD{\M[2,3]C\stepII \M[3]{W}} = 36$.
After row-reduction, we obtain $\M[2,3]{B} \M[3]{W}$:
\begin{align*}
    \M[2,3]{B} \preceq \left(\begin{matrix}
      17 & 14 & 10 & 6 \\
      17 & 13 & 10 & 7 \\
      16 & 13 & 9 & 7 \\
      16 & 13 & 10 & 6
  \end{matrix}\right) \quad \quad
    \M[2,3]{B} \M[3]{W} \preceq \left(\begin{matrix}
      17 & 17 & 16 & 15 \\
      17 & 16 & 16 & 16 \\
      16 & 16 & 15 & 16 \\
      16 & 16 & 16 & 15
  \end{matrix}\right).
\end{align*}
The next element in $\var C$ is $\var S_2$, so we construct $\M[2,4]C\stepII$ according to \eqref{eqn_stepII} and get:
\begin{align*}
    \M[2,4]C\stepII \preceq \left(\begin{matrix}
      17 & 14 & 10 & 6 & \bot \\
      17 & 13 & 10 & 7 & \bot \\
      16 & 13 & 9 & 7 & \bot \\
      16 & 13 & 10 & 6 & \bot \\
      \bot & \bot & 30 & 15 & 0
  \end{matrix}\right) & \quad \quad
    \M[2,4]C\stepII \M[4]{W} \preceq \left(\begin{matrix}
      17 & 17 & 16 & 15 & \bot \\
      17 & 16 & 16 & 16 & \bot \\
      16 & 16 & 15 & 16 & \bot \\
      16 & 16 & 16 & 15 & \bot \\
      \bot & \bot & 36 & 24 & 12
  \end{matrix}\right) \\
  \M[2,4]{B} \preceq \left(\begin{matrix}
      16 & 13 & 10 & 6 & 3 \\
      15 & 13 & 8 & 6 & 3 \\
      16 & 12 & 9 & 6 & 3 \\
      14 & 12 & 9 & 6 & 3 \\
      14 & 12 & 9 & 6 & 2
  \end{matrix}\right) & \quad \quad
  \M[2,4]{B} \M[4]{W} \preceq \left(\begin{matrix}
      16 & 16 & 16 & 15 & 15 \\
      15 & 16 & 14 & 15 & 15 \\
      16 & 15 & 15 & 15 & 15 \\
      14 & 15 & 15 & 15 & 15 \\
      14 & 15 & 15 & 15 & 14
  \end{matrix}\right).
\end{align*}
We needed 86 row reductions for the module minimisation, while the upper bound was 145 since we calculated $\OD{\M[2,4]C\stepII \M[4]{W}} = 24$.

The last iteration is again a $\var{Root}$, and we can use either of the two last rows of $\M[2,4]C\stepII$ since they have weighted degree $< s(n-\tau) = 16$.
Using the last, the obtained polynomial is:
\begin{align*}
Q_{2,4}(X,Y) & = (6X^3 + 16X^2 + 10X)Y^4 + (11X^6 + 16X^5 + 14X^4 + 15X^2 + 5X + 1)Y^3\\
& \quad + (15X^9 + X^8 + 5X^7 + 6X^6 + 12X^5 + 2X^4 + 6X^3 + 2X^2 + 12X + 2)Y^2 \\
& \quad + (5X^{12} + 5X^{11} + 2X^{10} + 9X^9 + 14X^8 + 6X^7 + 4X^6 + 3X^5 + 16X^4 \\
& \quad + X^3 + 16X^2 + 6X)Y + 7X^{14} + 16X^{13} + 6X^{12}+ 4X^{11} + 11X^{10} \\
& \quad + 11X^8 + 4X^7 + 5X^6 + 16X^5 + 12X^4 + 15X^3 + 6X^2 + 16X + 1.
\end{align*}
Indeed, $Q_{2,4}(X,2X^2+10X+6) = 0$ and root-finding retrieves $f(X)$ for us.

As the example shows, performing module minimisation on a matrix can be informally seen to ``balance'' the row-degrees such that they all become roughly the same size.
The complexity of this reduction depends on the number of row reductions, which in turn depends on the ``unbalancedness'' of the initial matrix.
The matrices $\M[1,2]C\stepI, \M[2,3]C\stepII$ and  $\M[2,4]C\stepII$ are more balanced in row-degrees than using $\M[1,2]{A}, \M[2,3]{A}$ and $\M[2,4]{A}$ directly.
\end{exa}

\subsection{Complexity Analysis} \label{ssec_complanalysis}
Using the estimates of the two preceding subsections, we can make a rather precise worst-case asymptotic complexity analysis of Algorithm~\ref{alg_multitrial}.
The average running time will depend on the exact choice of $\var C$ but we will see that the worst-case complexity will not.
First, it is necessary to know the complexity of performing a root-finding attempt.
\begin{lem}[Complexity of Root-Finding] \label{lem_rootCompl}
  Given a polynomial $Q(X,Y) \in \Fqx[Y]$ of $Y$-degree at most $\ell$ and $X$-degree at most $N$, there exists an algorithm to find all $\Fqx$-roots of complexity $O\big(\ell^2N\log^2 N\log\log N\big)$, assuming $\ell, q \in O(N)$.
\end{lem}
\begin{proof}
  We employ the Roth--Ruckenstein~\cite{rothRuckenstein00} root-finding algorithm together with the divide-and-conquer speed-up by Alekhnovich~\cite{alekhnovich05}.
  The complexity analysis in~\cite{alekhnovich05} needs to be slightly improved to yield the above, but see \cite{beelen_rational_2013} for easy amendments.\qed
\end{proof}

\begin{thm}[Complexity of Algorithm~\ref{alg_multitrial}] \label{thm_ComplAlg}
  For a given $\RS n k$ code, as well as a given list of steps $\var C$ for Algorithm~\ref{alg_multitrial} with ultimate parameters $(s, \ell, \tau)$, the algorithm has worst-case complexity $O(\ell^4 s n \log^2 n \log\log n)$, assuming $q \in O(n)$.
\end{thm}

\begin{proof}
  The worst-case complexity corresponds to the case that we do not break early but run through the entire list $\var C$.
  Precomputing $\M[1,1] A$ using Lagrangian interpolation can be performed in $O(n\log^2 n\log\log n)$, see e.g. \cite[p. 235]{gathen}, and reducing to $\M[1,1] B$ is in the same complexity by Lemma~\ref{lem_alekcompl}.
  
  Now, $\var C$ must contain exactly $\ell-s$ $\var S_1$-elements and $s-1$ $\var S_2$-elements.
  The complexities given in Corollaries~\ref{cor_ComplMSI} and \ref{cor_ComplMSII} for some intermediate $\isell$ can be relaxed to $s$ and $\ell$.
  Performing $O(\ell)$ micro-steps of type I and $O(s)$ of type II is therefore in $O(\ell^4 s n \log^2 n\log\log n)$.

  It only remains to count the root-finding steps.
  Obviously, it never makes sense to have two $\var{Root}$ after each other in $\var C$, so after removing such possible duplicates, there can be at most $\ell$ elements $\var{Root}$. 
  When we perform root-finding for intermediate $\isell$, we do so on a polynomial in $\Mod[\isell] M$ of minimal weighted degree, and by the definition of $\Mod[\isell] M$ as well as Theorem \ref{thm_GSproblem}, this weighted degree will be less than $\is(n-\itau) < sn$.
  Thus we can apply Lemma \ref{lem_rootCompl} with $N = sn$.
  \qed
\end{proof}
The worst-case complexity of our algorithm is equal to the average-case (and worst-case) complexity of the Beelen--Brander~\cite{beelenBrander} list decoder. However, Theorem~\ref{thm_ComplAlg} shows that we can choose as many intermediate decoding attempts as we would like without changing the worst-case complexity.
One could therefore choose to perform a decoding attempt just after computing $\M[1,1] B$ as well as every time the decoding radius has increased.
The result would be a decoding algorithm finding all \emph{closest} codewords within some ultimate radius $\tau$.
If one is working in a decoding model where such a list suffices, our algorithm will thus have much better average-case complexity since fewer errors occur more frequently than many.

\section{Re-Encoding Transformation} \label{sec_re-encoding}

\newcommand{\RMod}{\bar{R}}
\newcommand{\BasisReEnc}{\M[\sell]{\bar A}}
\newcommand{\GReEnc}{\bar G}

We now discuss how to adjust Algorithm~\ref{alg_multitrial} to incorporate the re-encoding transformation proposed in~\cite{koetter11,kotter_complexity_2003}.
The basic observation is that we can correct $\vec r$ if we can correct $\vec r - \hat {\vec c}$ for any $\hat {\vec c} \in \RS n k$. If we chose $\hat {\vec c}$ such that $\vec r - \hat {\vec c}$ for some reason is easier to handle in our decoder, we can save computational work.
As in the original articles, we will choose $\hat {\vec c}$ such that it coincides with $\vec r$ in the first $k$ positions; this can be done since it is just finding a Lagrange polynomial of degree $k-1$ that goes through these points.
The re-encoded received word will therefore have 0 on the first $k$ positions.

For ease of notation, assume that $\vec r$ is this re-encoded received word with first $k$ positions zero, and we can reuse all the objects introduced in the preceding sections.
Define 
\begin{equation} \label{def_Langrangian}
L(X) \defeq \prod_{i=0}^{k-1} (X-\alpha_i).
\end{equation}
Obviously $L(X) \mid G(X)$ so introduce $\GReEnc (X) = G(X)/L(X)$.
However, since $r_i = 0$ for $i < k$ then also $L(X) \mid R(X)$; this will be the observation which will save us computations.
Introduce therefore $\RMod(X) = R(X)/L(X)$.
Regard now $\M[\sell] A$ of \eqref{eq_BasisOfModul}; it is clear that $L(X)^{s-t}$ divides every entry in the $t$th column for $t < s$.
This implies that the image of the following bijective map is indeed $\Fqxy$:
\begin{align} \label{eq_MapRe_Enc}
\varphi: \quad \Mod[\sell] M \quad & \rightarrow  \quad \Fqxy \nonumber \\
Q(X,Y) \quad & \mapsto \quad L(X)^{-s} Q(X, L(X)Y).
\end{align}
Extend $\varphi$ element-wise to sets of $\Mod[\sell] M$ elements, and note that $\varphi$ is therefore an isomorphism between $\Mod[\sell] M$ and $\varphi(\Mod[\sell] M)$.
The idea is now that the elements in $\varphi(\Mod[\sell] M)$ have lower $X$-degree than those in $\Mod[\sell] M$, and we can therefore expect that working with bases of $\varphi(\Mod[\sell] M)$ is computationally cheaper than with bases of $\Mod[\sell] M$.
Since we are searching a minimal $(1,k-1)$-weighted polynomial in $\Mod[\sell] M$, we need to be sure that this property corresponds to something sensible in $\varphi(\Mod[\sell] M)$.
The following lemma and its corollary provides this:
\begin{lem} \label{lem_Wdeg_ReEnc}
For any $Q(X,Y)$ in $\Mod[\sell]{M}$
\begin{equation*}
\wdeg{1}{k-1} Q(X,Y) = \wdeg{1}{-1} \varphi(Q(X,Y)) + sk.
\end{equation*}
\end{lem}
\begin{proof}
We have $\wdeg{1}{k-1} Q(X,Y) = \max_i \big\{\deg \yC Q i (X) + i(k-1) \big\}$ so we obtain: 
\begin{align*}
 \wdeg{1}{-1} \varphi(Q(X,Y)) & = \max_i \big\{ \deg \yC Q i (X) - s \deg L(X) + i \deg L(X) - i \big\} \\
 & = \max_i \big\{ \deg \yC Q i (X) + i(k-1) \big\}- sk.
\end{align*}%
\qed
\end{proof}

\begin{cor} \label{cor_Wdeg_ReEnc}
$Q(X,Y)$ has minimal $(1,k-1)$-weighted degree in $\Mod[\sell]{M}$ if and only if $\varphi(Q(X,Y))$ has minimal $(1,-1)$-weighted degree in $\varphi(\Mod[\sell]{M})$.
\end{cor}
Let us now describe the basis of $\varphi(\Mod[\sell]{M})$ corresponding to the one in~\eqref{eq_BasisOfModul}:
\begin{thm} \label{thm_Mbasis_ReEnc}
The module $\varphi(\Mod[\sell]{M})$ is generated as an $\Fqx$-module by the $\ell+1$ polynomials $\bar P\T t(X,Y) \in \Fqxy$ given by
\begin{align*}
  \bar P\T t(X,Y) &= \GReEnc(X)^{s-t}(Y-\RMod(X))^t,         && \textrm{for } 0 \leq t < s,\\
  \bar P\T t(X,Y) &= (L(X)Y)^{t-s}(Y-\RMod(X))^s,            && \textrm{for } s \leq t \leq \ell.
\end{align*}
\end{thm}
\begin{proof}
Follows directly from Theorem~\ref{thm_Mbasis} and the mapping as defined~\eqref{eq_MapRe_Enc}.
\end{proof}
We can represent the basis of $\varphi(\Mod[\sell]{M})$ by the $(\ell+1) \times (\ell+1)$ matrix over $\Fqx$ (compare to~\eqref{eq_BasisOfModul}):
\begin{equation} \label{eq_BasisOfModul_ReEnc}
\BasisReEnc \defeq
  \left(
    \begin{array}{ccccccccc}
\GReEnc^s                        &                            &                            &                           &                           &                                                      &  &     \\
\GReEnc^{s-1}(-\RMod)            & \GReEnc^{s-1}              &                            &                           &                           & \multicolumn{2}{c}{\multirow{2}{*}{\mbox{\Large 0}}} &        \\
\GReEnc^{s-2}(-\RMod)^2     & 2 \GReEnc^{s-2}(-\RMod)            &  \GReEnc^{s-2}                          &                           &                           & &        \\
\vdots                           &                            &                     \ddots &                            &                           &                                                      &        \\
(-\RMod)^s                       & \binom{s}{1}(-\RMod)^{s-1} & \dots                      & 1                         &                           &                                                      &        \\
                                 & L(-\RMod)^{s}              & \multicolumn{2}{c}{\dots}  & \hspace*{-2em}L                         &                           &                                                               \\
                                 &                            & L^2 (-\RMod)^{s}           & \multicolumn{2}{c}{\dots} & L^2                       &                                                      &        \\
\multirow{2}{*}{\mbox{\Large 0}} &                            & \multicolumn{2}{c}{\ddots} &                           &                           &                                                      & \ddots \\
                                 &                            &                            & L^{\ell-s}(-\RMod)^{s}    & \multicolumn{2}{c}{\dots} &                                                      &  & L^{\ell-s} 
    \end{array}
  \right).
\end{equation}
We need an analogue of Corollary \ref{cor_WeightSol} for the $(1,-1)$-weighted degree, i.e.,~we should find a diagonal matrix to multiply on $\BasisReEnc$ such that when module minimising the result, we will have a row corresponding to a polynomial in $\varphi(\Mod[\sell] M)$ with minimal $(1,-1)$-weighted degree.
We cannot use $\diag(1, X^{-1},\ldots,X^{-\ell})$, since multiplying with negative powers of $X$ might cause us to leave the polynomial ring; however, we can to this add the same power to all the diagonal elements such that they become non-negative:
\begin{equation} \label{eq_WeightedReEnc}
\M[\ell]{\bar W}  = \diag \Big(X^{\ell},X^{\ell-1}, \dots, 1 \Big).
\end{equation}
Therefore, for a vector $\vec q = (Q_0(X)X^\ell,\ldots,Q_\ell(X)X^0)$ in the row-space of $\BasisReEnc \M[\ell]{\bar W}$ corresponds a polynomial $Q(X,Y) = \sum_{t=0}^\ell Q_t(X) Y^t$, and we will have the identity $\deg \vec q = \wdeg{1}{-1} Q + \ell$.
Obviously then, a minimal degree vector in $\BasisReEnc \M[\ell]{\bar W}$ is a minimal $(1,-1)$-weighted polynomial in $\varphi(\Mod[\sell] M)$.

Finally, we need adjusted variants of the micro-steps I and II.
The necessary adaptions of Algorithm~\ref{alg_multitrial} are summarised in the following lemma:
\begin{lem} \label{lem_MicroStepsReEnc}
Let $\M {\bar B} \in \Fqx^{(\ell+1) \times (\ell+1)}$ be the matrix representation of a basis of $\varphi(\Mod[s,\ell] M)$.
If $\M {\bar B}  \M[\ell]{\bar W}$ is in weak Popov form, then one of the rows of $\M {\bar B}$ corresponds to a polynomial in $\varphi(\Mod[s,\ell] M)$ with minimal $(1,-1)$-weighted degree.

Modified micro-steps of type I and II can be obtained from the following.
Let $\bar B\T 0(X,Y), \ldots, \bar B\T \ell(X,Y) \in \Fqxy$ be a basis of $\Mod[\sell]{M}$.
Then the following is a basis of $\Mod[\sellI] M$:
\begin{equation}
  \label{eqn_reencB1}
  \bar B\T 0(X,Y),\ \ldots\ ,\ \bar B\T \ell(X,Y), (L(X)Y)^{\ell-s+1}(Y-\RMod(X))^s.
\end{equation}
Similarly, the following is a basis of $\Mod[\sellII]M$:
\begin{equation}
  \label{eqn_reencB2}
\GReEnc^{s+1}(X),\ \bar B\T 0(X,Y)(Y-\RMod(X)),\ \ldots\ ,\ \bar B\T \ell(X,Y)(Y-\RMod(X)).
\end{equation}
\end{lem}
\begin{proof}
  The first part follows from the previous discussion and analogously to Corollary \ref{cor_WeightSol}.
  The recursive bases of \eqref{eqn_reencB1} and \eqref{eqn_reencB2} follow completely analogous to Lemmas \ref{lem_MicroStepOne} and \ref{lem_MicroStepTwo}, given Theorem \ref{thm_Mbasis_ReEnc}.\qed
\end{proof}

\begin{exa}
In the case of the \RS{16}{4} code over $\F{17}$ with final decoding radius $\tau=8$ as shown in Example~\ref{ex_algo164}, then $\deg L(X) = 4$ and the initial matrix satisfies:
\begin{align*}
\M[1,1]{\bar A} = \left(\begin{array}{rr}
                \bar G & 0 \\
                -\bar R & 1
              \end{array}\right) \quad
\M[1,1]{\bar A} \preceq \left(\begin{array}{rr}
      12 & \bot \\
      11 & 0
  \end{array}\right)  \quad
\M[1,1]{\bar A} \M[1]{\bar W} \preceq \left(\begin{array}{rr}
      13 & \bot \\
      12 & 0
  \end{array}\right).
\end{align*}
\end{exa}

\begin{rem}
  Brander briefly described in his thesis \cite{brander10} how to incorporate re-encoding into the Beelen--Brander interpolation algorithm by dividing out common powers of $L(X)$ in the first $s$ columns of $\M[\sell] A$.
  Here we construct instead $\M[\sell]{\bar A}$ where powers of $L(X)$ are also multiplied on the latter $\ell-s$ columns, since we need the simple recursions of bases of $\varphi(\Mod[\sell] M)$ which enables the micro-steps.

  However, before applying module minimisation, we could divide away this common factor from those columns and just adjust the weights accordingly (i.e., multiplying $X^{k(t-s)}$ on the $t$th element of $\M[\ell]{\bar W}$); this will further reduce the complexity of the minimisation step.
  The micro-steps would then need to be modified; the simplest way to repair this is to multiply back the powers of $L(X)$ before applying a micro-step, and then remove them again afterwards.
  With a bit more care one can easily do this cheaper, though the details become technical.
\end{rem}
In asymptotic terms, the computational complexity of the iterative interpolation method stays exactly the same with re-encoding as without it, since $O(n - \deg L) = O(n-k) = O(n)$ under the usual assumption of $n/k$ being constant.
The same is true for the original re-encoding scheme of K\"otter--Vardy~\cite{koetter03,koetter11}.
However, most of the polynomials that are handled in the matrix minimisation will be of much lower degree than without re-encoding; for relatively high-rate codes this will definitely be noticeable in real computation time.

In~\cite[Thm. 10]{koetter11}, it was shown that the root-finding procedure of Roth--Ruckenstein~\cite{rothRuckenstein00}, or its divide-\&-conquer variant by Alekhnovich~\cite{alekhnovich05} can be directly applied to an interpolation polynomial in $\varphi(\Mod[s,\ell] M)$, 
so we can avoid to construct and work on the larger polynomial in $\Mod[s,\ell] M$.
Instead of finding $f(X)$, one will find the power series expansion of $f(X)/L(X)$.
The fraction in reduced form can be retrieved from the power series expansion using Pad\'e approximation: e.g.~by the Berlekamp--Massey algorithm or by module minimising a certain $2\times 2$ matrix.
See e.g.~\cite[Section 2.5]{nielsenThesis} for a general description of the latter.
From the reduced fraction, $f(X)$ can be obtained by re-extending the fraction.

Interestingly, one can easily calculate that the orthogonality defects stays the same, i.e.,~$\OD{\M[\sellI] C\stepI \M[\ell+1]W} = \OD{\M[\sellI] {\bar C}\stepI \M[\ell+1]{\bar W}}$, where $\M[\sellI] {\bar C\stepI}$ is the matrix corresponding to a micro-step of type I in the re-encoded version.
The analogue equality holds for type II.
This means that, roughly, the number of row operations carried out by the module minimisation algorithm is unchanged.

\section{Simulation Results} \label{sec_simulation}
The proposed algorithm has been implemented in Sage, Version 5.13~\cite{sage2013}, using the Mulders--Storjohann algorithm for module minimisation, and the Roth--Ruckenstein root-finding procedure \cite{rothRuckenstein00}.
For comparison, we also implemented construction of $\M[\sell] A$, leading immediately to the Lee--O'Sullivan algorithm~\cite{leeOSullivan08}.

Figure~\ref{fig_sim164} shows the total number of finite field multiplications performed for complete runs of the decoding algorithms, using the \RS{16}{4} code over $\F{17}$ as considered in Example~\ref{ex_rs164param}.
For each algorithm, and for each number of errors $\varepsilon \leq \tau$, 1000 random codewords were generated  and subjected to a random error pattern of weight precisely $\varepsilon$.
The solid line gives the number of operations of the proposed multi-trial algorithm for any number of errors.
For the Lee--O'Sullivan decoder, one chooses the maximal decoding radius initially, and the figure depicts choosing both $\tau=7,8$ as dashed lines.
The minimum-distance choice of $\tau=6$ coincides completely with the multi-trial algorithm since $\M[\sell] A$ for $\varepsilon \leq 6$ and so is not shown.

The figure demonstrate that the multi-trial algorithm provides a huge gain whenever there are fewer errors than Lee--O'Sullivan's target, while not having a disadvantage in the matching case.
The right-hand graph shows the complexity when using the re-encoding transformation, and we can observe a speedup of between 30\% and 50\%.
\begin{figure}[htb]
 \centering 
 \hspace*{-1em}%
 \subfigure[Without Re-Encoding]{\includegraphics[width=.505\textwidth]{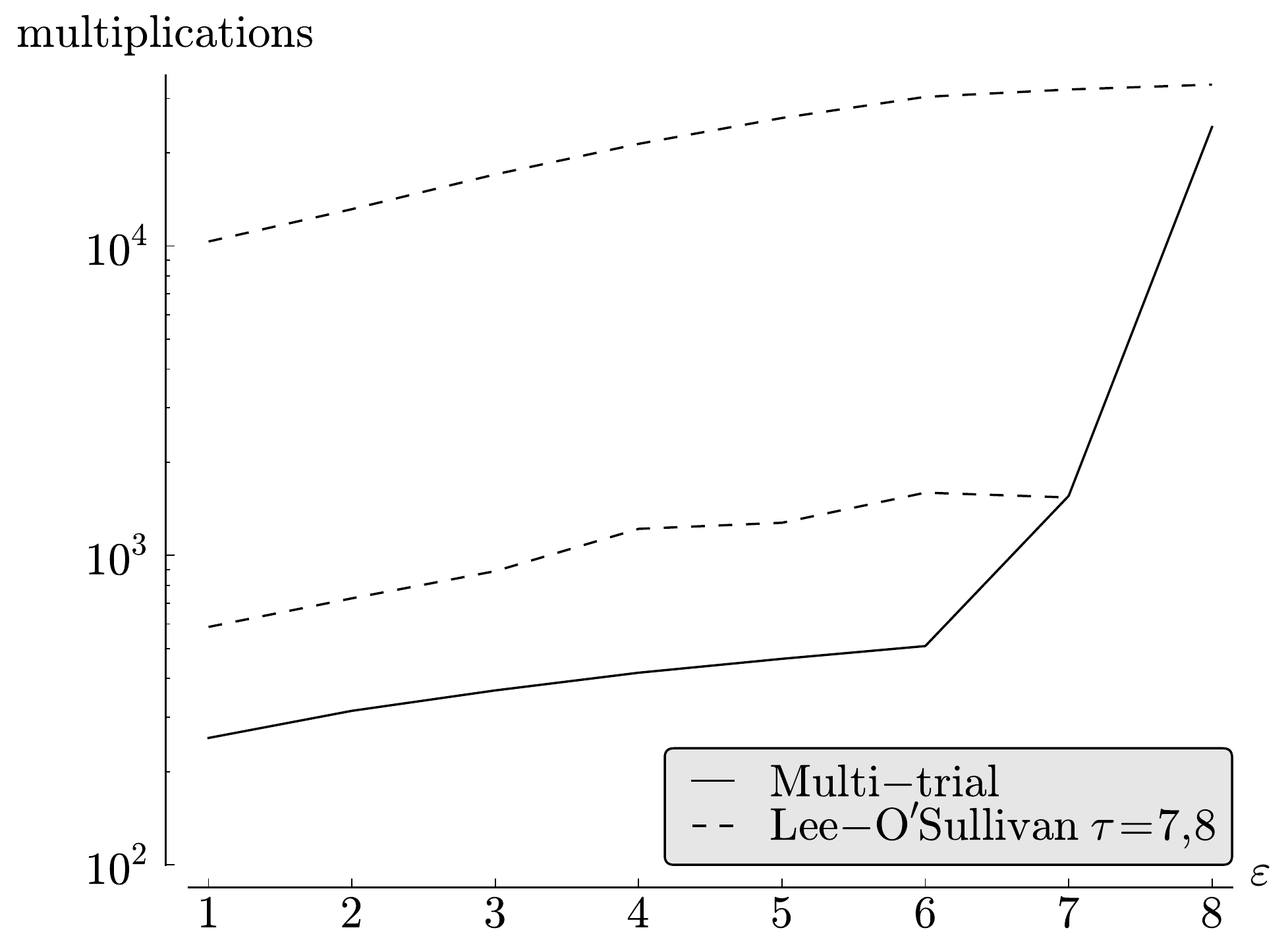}}
 \subfigure[With Re-Encoding]{\includegraphics[width=.505\textwidth]{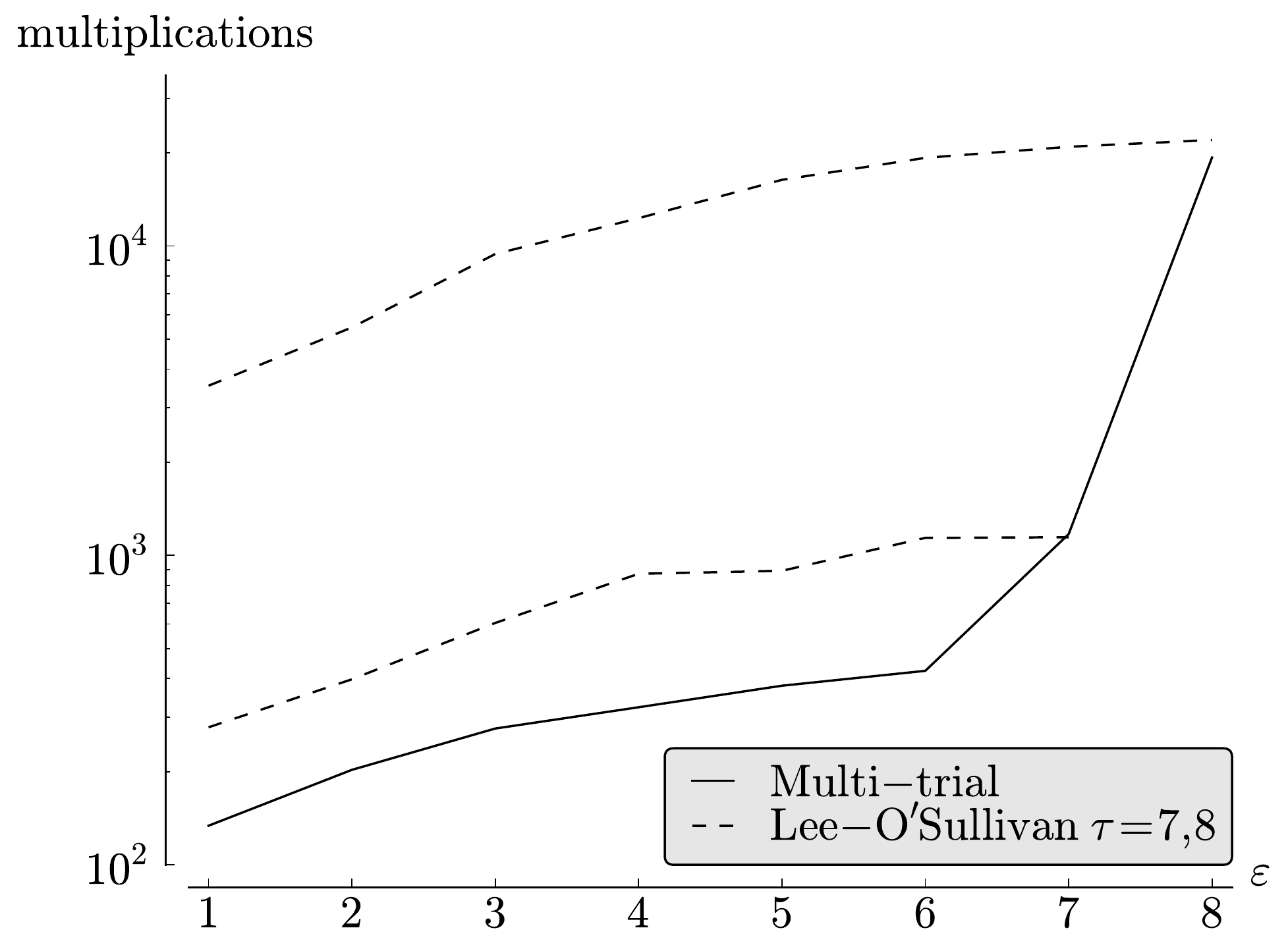}}%
 \hspace*{-1em}
 \caption{Comparison of the number of field multiplications for list decoding of an \RS{16}{4} code over $\F{17}$. The operations of the Lee--O'Sullivan~\cite{leeOSullivan08} algorithm with different aimed decoding radii and the proposed multi-trial algorithm are illustrated.
   Note the logarithmic $y$-axis.
   Subfigure (a) illustrates the number of operations without re-encoded points. The gain due to the re-encoding procedure is visible in Subfigure (b).}
\label{fig_sim164}
\end{figure}
The number of operations spent in constructing the matrices and for the root-finding step are relatively small compared to the number of operations needed for the row-reduction (for the \RS{16}{4} code, less than 5\%). 

Unfortunately, the performance of our operation-counting implementation for the decoding algorithm does not allow to run simulations with much larger codes.
Without counting, however, we can use optimised data structures in Sage~\cite{sage2013}
which run much faster. We have observed for the short \RS{16}{4} code that the system-clock time spent with these data structures correspond very well to the operations
counted. Extrapolating, we can, albeit with larger uncertainty, make comparisons
on larger codes using system-clock measurements.
Doing so, we have observed behaviour resembling that of the \RS{16}{4} code.
For example, decoding a \RS{64}{25} code up to 23 errors (requiring $(s,\ell) = (4,6)$) showed the multi-trial being slightly faster than Lee--O'Sullivan in the worst
case, while of course still giving a large improvement for fewer errors.
When decoding a \RS{255}{120} code up to 74 errors (requiring $(s,\ell) = (4,5)$), Lee--O'Sullivan~\cite{leeOSullivan08} was slightly faster by about 15\% in the worst case.
It should also be noted that root-finding took up significantly more time for these larger codes, around 15\% of the total time for Lee--O'Sullivan and 20\% for the multi-trial.

As with any simulation, there are caveats to these results.
In truth, only the wall-clock time spent by highly optimised implementations of various approaches can be fairly compared.
Also, we did not test the asymptotically fast Alekhnovich's method for module minimisation.
Investigations performed by Brander in Magma indicated that the gain in using this algorithm in place of Mulders--Storjohann might only be present once the code length exceeds about 4000 \cite{brander10}.

The implementation of the algorithm, including the simulation setup, is freely available via \url{http://jsrn.dk/code-for-articles}.

\section{Conclusion} \label{sec_concl}
An iterative interpolation procedure for list decoding GRS codes based on Alekhnovich's module minimisation was proposed and shown to have the same worst-case complexity as Beelen and Brander's \cite{beelenBrander}.
We showed how the target module used in Beelen--Brander can be minimised in a progressive manner, starting with a small module and systematically enlarging it, performing module minimisation in each step.
The procedure takes advantage of a new, slightly more fine-grained complexity analysis of Alekhnovich's algorithm, which implies that each of the module refinement steps runs fast.

We furthermore incorporated the re-encoding transformation of K\"otter and Vardy \cite{koetter03} into our method, which provides a noticeable, if not asymptotic, gain in computational complexity.

The main advantage of the algorithm is its granularity which makes it possible to perform fast multi-trial decoding: we attempt decoding for progressively larger decoding radii, and therefore find the list of codewords closest to the received.
This is done without a penalty in the worst case but with an obvious benefit in the average case.

The Beelen--Brander approach for interpolation is not the asymptotically fastest: using the module minimisation algorithm by Giorgi~\textit{et al.}~\cite{giorgi03}, one gains a factor $\ell$. 
By a completely different approach, Chowdhury et al.~\cite{chowdhury_faster_2014} further beat this by a factor $\ell/s$, achieving $O(\ell^2s^2 n \log^{O(1)}(n))$.
It is unclear for which sizes of the parameters these asymptotic improvements have concrete benefits, and whether a multi-trial approach can be developed for them.

\subsection*{Acknowledgement}
The authors thank Daniel Augot for fruitful discussions. This work has been supported by German Research Council “Deutsche Forschungsgemeinschaft” (DFG) under grant Bo867/22-1 and Ze1016/1-1.

\end{document}